\documentclass[runningheads]{llncs}
\usepackage{tikz}

\usepackage[T1]{fontenc}
\usepackage{algorithm}
\usepackage{mathtools}
\usepackage{amsmath}
\usepackage{amssymb}
\usepackage{color}
\usepackage{graphicx}
\usepackage{gastex}
\usepackage{mathbbol}
\usepackage[hidelinks]{hyperref}
\usepackage{soul}
\usepackage{comment}

 \newcommand{\red}[1]{{\color{red}{ #1}}}
 \newcommand{\pf}[1]{\operatorname{Pref}(#1)}

\usetikzlibrary{arrows,automata,shapes,positioning}
\tikzset{every state/.style={minimum size=0pt}}

\begin{document}
\title{Finite-State Transducers in the Wheeler Setting}
%
%
\author{Giovanna D'Agostino \inst{1}\orcidID{0000-0002-8920-483X} \and \\Andrea Paradiso\inst{1}\orcidID{0000-0002-3614-2487} 
}
\authorrunning{G. D'Agostino and A. Paradiso}
%
\institute{University of Udine (DMIF), Italy}
\maketitle       
%

%
%

\section{Introduction}



In this work  we take a preliminary step toward understanding the connections between  Wheeler automata and finite-state transducers. 

Wheeler graphs form a class of edge-labeled directed graphs that admit a sorting of their nodes that is consistent with the co-lexicographic order of the labels of paths ending at them. A finite-state automaton whose underlying graph is a Wheeler graph is called a Wheeler automaton, and the languages recognized by such automata are called Wheeler languages~\cite{alanko_wheeler_2021}.  Wheeler languages form a proper subclass of the regular languages; in fact, they are a subclass of the star-free languages. They have attracted considerable interest thanks to structural properties enabling highly efficient indexing, pattern matching and compression~\cite{10.1093/bib/bbw089,backurs2016regularexpressionpatternshard,GAGIE201767}. They provide a formal automata-theoretic foundation for data structures inspired by the Burrows--Wheeler transform, bridging language theory and practical text indexing.
For this reason, Wheeler languages are important both as a mathematically natural subclass of regular languages and as a model for highly efficient algorithms on strings, graphs, and automata.


A natural way to deepen our understanding of a language class is to study the transformations that preserve it, either under direct image or under inverse image. This perspective has proved fruitful in the regular setting, where rational functions definable by  finite-state   subsequential transducers   (FSTs, \cite{filiot_transducers_2016})  are known to preserve regularity under both images and inverse images, and where algebraic tools, such as syntactic invariants associated with transductions, provide structural information about the transformations involved.   Subsequential transducers are a fundamental model for deterministic input-to-output transformations, capturing computations in which the output is produced online while the input is read from left to right.
They are important both theoretically, as a well-behaved subclass of finite-state transductions, and practically, as efficient models for tasks such as text processing, coding, normalization, and stream transformation (see~\cite{muscholl_many_2019} for a recent survey).

Identifying classes of transducers that are compatible with the Wheeler structure may therefore help clarify the boundaries of Wheeler languages and provide a more dynamic, transformation-oriented view of this class.  In this work we define a subclass of  transducers and argue that 
it is a natural class of transducers to consider in the setting of Wheeler languages. What properties should such a class have? 
We would certainly want it to be closed under composition, and, with respect to Wheeler languages, we would want it to be at least Wheeler-continuous, in the sense that the preimage of a Wheeler language for a transducer in the class should remain a Wheeler language \footnote{see \cite{Boj} for an  in-depth  discussion on the relevance of continuity for transducers}. In this work we  indeed prove that our class is closed under composition and that the computed function are  Wheeler-continuous. Moreover we give a machine-independent characterization of the functions computed by the class.

\section{Preliminaries}
We first fix some conventions before introducing Wheeler transducers in detail. Throughout the paper, whenever an alphabet $\Sigma$ is given, it is assumed to be equipped with an order denoted by $\prec$. Moreover, given two strings $\alpha = a_1a_2\dots a_n$ and $\beta = b_1b_2\dots b_m$ in $\Sigma^*$, their co-lexicographic order is defined by overloading the symbol $\prec$: 
  \begin{align*}
\alpha \prec \beta \iff & \bigl(n < m \land (\forall j < n)(a_{n-j} = b_{m-j})\bigr) \lor \\
&(\exists i)(a_{n-i}\prec b_{m-i} \land (\forall j<i)(a_{n-j} = b_{m-j}))
  \end{align*}

  \subsection{Automata and Transducers}
  We assume that the notions of deterministic (DFA)/nondeterministic (NFA) automata are known to the reader. Given a deterministic finite-state automaton $A = (Q,s,\Sigma,\delta, F)$   over an ordered alphabet $(\Sigma, \prec)$, a natural partial order on its states can be defined in terms of the strings reaching them. For each state $q \in Q$, let $I_q = \{\alpha \in \Sigma^* \mid \delta(s,\alpha) = q\}$ denote the set of strings reaching $q$. The co-lexicographic partial order $<_A$  on the set of states  of $A$ is defined as follows: for states $p,q \in Q$, we set $$p <_{A}q \iff \forall \alpha \in I_p \forall \beta \in I_q (\alpha\prec \beta ) $$ A \emph{Wheeler automaton} is then defined as a deterministic finite-state automaton in which the order $<_{A}$ is a total order\footnote{There is also a notion of nondeterministic Wheeler automata, which we omit for space reasons}.

We next consider the transducer notion. A one-way \emph{subsequential transducer} (which we simply call a transducer in this paper) computes a partial function from $\Sigma^*$ to $\Gamma^*$, where $\Sigma$ and $\Gamma$ are the input and output alphabets, respectively.
 It is defined as a tuple
$T = (Q, s,  \Sigma, \Gamma, \delta, \omega, F,t)$,
where:
\begin{enumerate}
    \item $A_T = (Q, s,  \Sigma, \delta,F)$ is the underlying DFA,
    \item the partial function $\omega : Q \times \Sigma \to \Gamma^*$ has the same domain as $\delta$,
    \item  $t: F \to \Gamma^*$ is the terminal function, which associates with each final state a string to be concatenated to the output produced while reading an accepted input string.
\end{enumerate} 
We call  $\delta$ and $ \omega$  the \emph{input} and \emph{output} transition functions, respectively, and we    extend them   as usual to $Q\times \Sigma^*$. If $\delta(q,a)=p$ and $\omega(q,a)=\gamma$ we also say that the transition $(q,a,\gamma,p)$ belongs to $T$.  

Given a   transducer $T$, we denote by $||T||: \Sigma^* \rightarrow \Gamma^*$ the partial function \emph{computed by } $T$, that is, the function $\alpha \mapsto \omega(s, \alpha) t(\delta(s,\alpha))$,   where $\alpha\in \Sigma^* $ is accepted by the underlying DFA $A_T$  (see~\cite{muscholl_many_2019} for  formal definitions).  A function  $f: \Sigma^* \rightarrow \Gamma^*$ is called a \emph{subsequential} function if there exists a transducer $T$ such that $||T||=f$. 
In the following, we always assume w.l.o.g. that the underlying DFA $A_T$ of a transducer $T$ is trimmed, that is, any state is accessible from the initial state and reaches a final state, by reading input words.

\begin{example}
 For the transducer in Fig. \ref{exampleWheelerTrans} we have  \(dom(||T||)=aab(ab)^*d+cab(ab)^*d\), 
\( ||T||(aab(ab)^kd)=000(010)^k01, ~~||T||(cab(ab)^kd)=110 (110)^k11\).  
\end{example}

 \begin{figure}[t]
 \begin{center}
\begin{tikzpicture}[
  ->,
  node distance=1cm,
  every state/.style={circle, draw, minimum size=8mm},
  initial text={}
]

\node[state, initial] (q0) {$q_0$};
\node[state, above right=of q0] (q1) {$q_1$};
\node[state, right=of q1] (q2) {$q_2$};
\node[state, right=of q2] (q3) {$q_3$};
 \node[state, accepting] (q7) at (7,0) {$q_7$};
\node[state, below=2.1cm of q1] (q4) {$q_4$};
\node[state, right=of q4] (q5) {$q_5$};
\node[state, right=of q5] (q6) {$q_6$};
 \node[ right=of q7 ] (q8) {};
\path
(q0) edge node[left] { $a/0$} (q1)
(q1) edge node[above] {$a/0$} (q2)
(q2) edge[bend left=20] node[above] {$b/0$} (q3)
(q3) edge[bend left =20] node[below] {$a/01$} (q2)
 (q3) edge node[above] {$d/0$} (q7)
(q0) edge node[left] { $c/1$} (q4)
(q4) edge[ ] node[above] {$a/1$} (q5)
(q5) edge[bend left=20] node[above] {$b/0$} (q6)
(q6) edge[bend left=20] node[below] {$a/11$} (q5)
(q7)edge node[above] {$ 1$} (q8) 
 (q6) edge node[above] {$d/1$} (q7);
\end{tikzpicture}
\caption{An example of a   (Wheeler) transducer. The terminal function $t$ is defined as: $t(q_7) = 1$}
  \label{exampleWheelerTrans}
\end{center}
\end{figure}
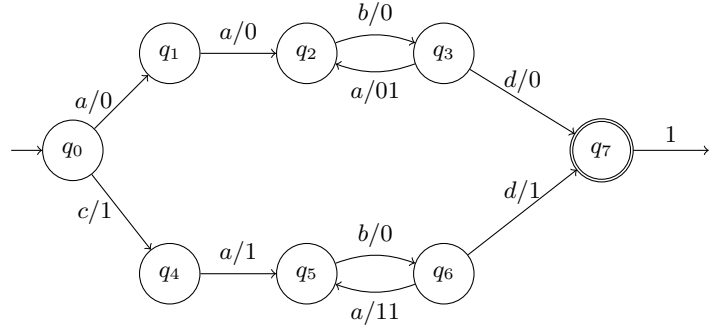

Two transducers $T,T'$ are equivalent if they compute the same function. 
As proved in \cite{CHOFFRUT2003131}, the  class of transducers introduced above enjoys a minimization property. Given a partial function $f:\Sigma^* \rightarrow \Gamma^*$, Choffrut defined   a syntactic congruence $\sim_f$  with the property that $f$ is computed by a   transducer if and only if $\sim_f$ has finite index. 
 Moreover, the equivalence $\sim_f$ allows to define a "minimum" transducer computing the function $f$. 
 In order to define the congruence $\sim_f$, one first defines a   function $\widehat f:\Sigma^* \rightarrow \Gamma^*$ which, given a word $\alpha$, outputs the longest common prefix of all words $f (\alpha \beta)$ for all continuations
$\beta$ with $\alpha \beta \in dom (f)$. 
More precisely, for \(\alpha\in \pf{dom(f)}\), let
\[ \widehat f(\alpha)= \bigwedge_{\alpha \beta \in dom (f) }f(\alpha\beta)
\]
Then $\widehat f$ is used to define the following equivalence:  
\begin{definition}\label{def:synct_congruence} The syntactic congruence \(\sim_f \)  of a partial function  $f:\Sigma^* \rightarrow \Gamma^*$  is  defined on $\pf{dom(f)}$   by:
\[
    \alpha \sim_f \alpha' \Leftrightarrow \forall \beta   \left ( \widehat f (\alpha)^{-1} f(\alpha \beta) =\widehat f (\alpha')^{-1} f(\alpha' \beta)\right )
\]
\end{definition} 
 
\begin{remark}  \label{rem:eq_convex}
The equality on the right must be seen as an equality between partial function, so that if $\alpha \sim_f \alpha'$ then, for all $\beta \in \Sigma^*$ we have  $\alpha\beta \in dom(f) \iff  \alpha'\beta \in dom(f)$. 
Moreover, if \(\alpha\sim_f\alpha'\) and \(\alpha\beta, \alpha'\beta \in {dom}(f)\),
then   there exists
\(x\in\Gamma^*\) such that
\[
f(\alpha\beta)=\widehat f(\alpha)x
\qquad\text{and}\qquad
f(\alpha'\beta)=\widehat f(\alpha')x.
\]
\end{remark}

Using the equivalence $\sim_f$,  the (possibly infinite state)  transducer $T_f$ associated with the partial function $f$ is then obtained as follows:

\begin{definition}\label{def:Choffrut}
   If $f:\Sigma^*\rightarrow \Gamma^*$  is a partial function then  
\[
T_f=(Q_f,s_f,\delta_f,\omega_f,F_f,t_f),
\]
where $s_f$ is a new symbol and 
\begin{itemize}
    \item $Q_f=\{s_f\} \cup \{[\alpha]_f : \alpha\in\operatorname{Pref}(\operatorname{dom}(f))\}$;
    \item $\delta_f(s_f,a)=[a]_f$,    if $a\in \pf{dom(f)}$;
    \item $\delta_f([\alpha]_f,a)=[\alpha a]_f$,  if $\alpha a\in \pf{dom(f)}$;
    \item $\omega_f(s_f,a)=  \widehat f (a)$,   if $a\in \pf{dom(f)}$;
    \item $\omega([\alpha]_f,a)=\widehat f(\alpha)^{-1}\widehat f(\alpha a)$,   if $\alpha a\in \pf{dom(f)}$;
    \item   $F_f=\{[ \alpha]_f : \alpha\in\operatorname{dom}(f)\}\cup \begin{cases}
        \{s_f\}~\text{if}~ \epsilon \in dom(f);\\
        \emptyset ~~~~\text{otherwise.}
    \end{cases}$
    \item $t_f([\alpha]_f)=\widehat f(\alpha)^{-1}f(\alpha)$;
    \item  $t_f(s_f)=\begin{cases} f(\epsilon),~~~~~~~\text{if}~ \epsilon\in dom(f);\\
    \text{undefined,  otherwise.}
    \end{cases} $
\end{itemize}
\end{definition}

\begin{lemma} (see \cite{CHOFFRUT2003131})
    If $f:\Sigma^*\rightarrow \Gamma^*$  is a partial function then  
\(
T_f 
\)    is well defined and realizes the function $f$.
\end{lemma}
\begin{proof} (Sketch)
First, the 
right invariance of $\sim_f$ is proved and this implies  that $\delta_f$ is well defined; then the definition of $\sim_f$ implies that $\omega_f$ and $t_f$ are well defined.

Let $\alpha =a_1\cdots a_n$ and let $q_i=[a_1\cdots a_i]_f$ for $i=1,\dots,n$.
The output produced by $T_f$ on $\alpha$ is
\begin{align*}
||T_f||(\alpha) &=
\omega_f(s_f,a_1)\omega_f(q_1,a_2)\cdots \omega_f(q_{n-1},a_n)\,t_f(q_n)\\
&=
\widehat f(a_1)\,
(\widehat f(a_1)^{-1}\widehat f(a_1a_2))\cdots
(\widehat f(a_1\cdots a_{n-1})^{-1}\widehat f(\alpha))\,
(\widehat f(\alpha)^{-1}f(\alpha)) \\
&= f(\alpha).
\end{align*}

since all the intermediate factors cancel. Hence $\|T_f\|(\alpha)=f(\alpha)$ for every $\alpha\in\Sigma^*$.  \hfill $\Box$
\end{proof}

Given a transducer $T$, the equivalence $\sim_T$ on $\pf{dom(||T||)}$ is defined as  follows
\begin{equation}\label{eq:simT}
 \alpha \sim_T \alpha' \iff \delta_T(s,\alpha)=\delta_T(s,\alpha')   
\end{equation}

Then it is easy to see that if $f=||T||$ is the function computed by $T$ then 
$\sim_T$ is a refinement of $\sim_{f}$, that is 
\[
\alpha \sim_T \alpha' \implies \alpha \sim_f \alpha'.
\]

It follows:
\begin{corollary} A function $f$ is subsequential iff $\sim_f$ has finite index. 
    \end{corollary}
 
If $T$ is a transducer and    $\sim$ is a right-invariant equivalence of finite index  which is a  refinement of $\sim_T$ we consider  the following construction of a transducer $T_\sim$ based on the equivalence $\sim$:

\begin{definition}
    Let $T=(Q ,s , \delta , \omega , F , t )$ be a transducer. If $\sim$ is a right invariant equivalence, of finite index, which is a refinement of  $\sim_T$, then we define a  quotient  transducer $T_\sim=(Q_\sim, s_\sim,\delta_\sim,\omega_\sim,F_\sim,t_\sim)$ as follows:
\begin{itemize}
    \item the set of states $Q_\sim$ are the $\sim$-equivalence classes $[\alpha]_\sim$;
    \item $s_\sim=[\epsilon]_\sim$ is the initial state;
    \item $\delta_\sim([\alpha]_\sim,a)=[\alpha a]_\sim$;
    \item $\omega_\sim([\alpha]_\sim,a)=\omega(\delta (s , \alpha),a)$;
    \item $F_\sim=\{[\alpha]_\sim\; :\; \delta (s, \alpha)\in F \}$;
    \item $t_\sim([\alpha]_\sim)=t (\delta (s, \alpha))$.
\end{itemize} 
\end{definition}
In this paper we shall use the following result:
\begin{lemma} \label{lem:sim_quotient}$T=(Q ,s , \delta , \omega , F , t )$ be a transducer, let $\sim$ be a right invariant equivalence, of finite index, which is a refinement of  $\sim_T$.  Then the quotient transducer      $T_\sim$ is well defined  and  for all $\alpha, \beta$ with $\alpha \beta \in \pf{dom(||T||)}$ it holds: 
 \[
\omega_\sim([\alpha]_\sim, \beta)=\omega(\delta(s,\alpha), \beta) 
\]
Moreover, the transducer $T_\sim$ is equivalent to $T$. 
\end{lemma}
\begin{proof}
The functions $\delta_\sim $ is well defined because $\sim$ is right invariant. As for $ \omega_\sim$, given two strings $\alpha,\beta \in \pf{dom(||T||)}$ such that $[\alpha]_\sim = [\beta]_\sim$, since $\sim$ is a refinement of $\sim_T$, we have $\delta(s,\alpha) = \delta(s,\beta)$; hence 
$\omega_\sim([\alpha]_\sim,a) = \omega(\delta(s,\alpha),a) = \omega(\delta(s,\beta),a) = \omega_\sim([\beta]_\sim,a)$ holds. Similarly, we easily prove that $t$ and $F_\sim$ are well defined (in the sense that if $\delta(s,\alpha)\in F$, then $[\alpha]_\sim \subseteq F).$ 

Moreover, one can prove by induction on $\beta$ that 
 \[
\omega_\sim([\alpha]_\sim, \beta)=\omega(\delta (s,\alpha), \beta).
\]
  We now  prove that $||T||$ and  $||T_\sim||$ have the same domain and, given a word $\alpha \in dom(||T||)$, then $||T||(\alpha) = ||T_\sim||(\alpha)$. The first part follows from the definition of $F_\sim$. Moreover, we have: 
    \begin{align*}
    \|T_\sim\|(\alpha)
    &=\omega_\sim ([\epsilon]_\sim,\alpha)\,
      t_\sim (\delta_\sim ([\epsilon]_\sim,\alpha))\\
    &=\omega(s,\alpha)\,
      t(\delta(s,\alpha))\\
    &=\|T\|(\alpha).
    \end{align*}
    Therefore $T_\sim$ and $T$ are equivalent. \hfill $\Box$
\end{proof}
\subsection{Properties of  piecewise monotone functions}

As we shall see, our class of transducers relies on the notion of \emph{piecewise monotone function}. 

\begin{definition}\label{def:piecewise} A partial function 
  $f:\Sigma^*\rightarrow \Gamma^*$  is   \emph{piecewise monotone} if there is a finite   partition 
$\mathcal C=\{C_1, \dots,C_n\}$ of $dom(f)$ consisting of convex sets such that  $f$ is monotone, either nondecreasing or nonincreasing,  on each  $C_i$.
In this case we  say that the convex partition $\mathcal C$ \emph{realizes} the piecewise monotonicity of $f$. 
\end{definition}

\begin{lemma}\label{lem:comp_piecewise}
    Piecewise monotone functions are closed under composition 
\end{lemma}
\begin{proof}
Suppose $f:X\rightarrow Y$ and $g:Y\rightarrow Z$ are piecewise monotone w.r.t. convex $C_1,\dots,C_h$, $D_1, \dots, D_k$, respectively. 
Then, for all $i,j$, the sets  \(C_i\cap f^{-1}(D_j)\) are convex and   $g\circ f$ is piecewise monotone w.r.t. the non-empty set of the form \(C_i\cap f^{-1}(D_j)\)
    \end{proof}

\section{Wheeler Transducers}
We introduce the notion of deterministic finite-state Wheeler transducers and study their fundamental properties. We first show that the class is closed under composition, and we also prove that Wheeler languages are closed under inverse images of Wheeler transductions. We then turn to minimization: building on Choffrut's theory of minimal sequential transducers~\cite{CHOFFRUT2003131}, we introduce a refinement of the syntactic Choffrut equivalence, denoted $\sim_f^c$, and show that it characterizes exactly the functions realizable by a Wheeler  transducer. Finally, we provide a machine-independent characterization of   functions  computed by  Wheeler transducers  in terms of monotonicity.

  We first present our proposal for the class of Wheeler transducers.

  \begin{definition} \label{def:WheelerTransducers}
  A transducer $T = (Q, s, \Sigma, \Gamma, \delta, \omega, F, t)$ is said to be
a Wheeler transducer if, denoting by $<_{\mathrm{in}}$ the input order
$<_{A_T}$ on $Q$ induced by the underlying DFA
$A_T = (Q, s, \Sigma, \delta, F)$, the following conditions hold:
\begin{enumerate}
     \item the order $<_{\mathrm{in}}$ is total on $Q$;
\item   for every state $q \in Q$, the function $\omega(s,\alpha)$
   is piecewise monotone  on $I_q$.
    \end{enumerate} 
 If \(T\) is a Wheeler transducer, then the  function $||T||$ is called   \emph{a Wheeler function}.
\end{definition}

An example of a Wheeler transducer is shown in Fig.~\ref{exampleWheelerTrans}, where the 
co-lex order on states is $q_0 <_{\mathrm{in}} q_1 <_{\mathrm{in}} q_2 <_{\mathrm{in}} q_5<_{\mathrm{in}} q_3<_{\mathrm{in}} q_6<_{\mathrm{in}} q_4<_{\mathrm{in}} q_7$  (a total order) and one can verify that condition 2 of Definition \ref{def:WheelerTransducers} is satisfied  because, for each state $q$, the restriction of $\omega(s,\alpha)$ to $I_q$ is piecewise monotone. 
For instance,  $\omega(s,\alpha)$ is monotone decreasing on $I_{q_5}$, hence piecewise monotone with one piece, while it is piecewise monotone with two pieces on $I_{q_7}$.

\medskip

 \begin{remark}\label{rem:1}
    Given a DFA $D=(Q,s,\Sigma, \delta, F)$  we define  the \emph{identity  transducer $Id_D$ } based on $D$  as follows: 
  \[Id_D = (Q,s,\Sigma,\Sigma,\delta,Id,F, t),\] where $t(q)=\epsilon$, for all $q\in F$.
Then one can easily check that $D$ is a Wheeler DFA iff $Id_D$ is a Wheeler transducer.
\end{remark}

\subsection{Closure Properties}

\begin{definition}
    A function $f:\Sigma^*\rightarrow \Gamma^*$ is said to be \emph{Wheeler-continuous} if, for every Wheeler language $L\subseteq \Gamma^*$,   the inverse image $f^{-1}(L)$ is a Wheeler language.   
\end{definition}

We first prove that Wheeler function are Wheeler-continuous. 
 
\begin{lemma} If $T$ is a Wheeler transducer then $||T||$ is Wheeler-continuous.
\end{lemma}
\begin{proof}
Let
\(
L\subseteq \Gamma^*
\)
be a Wheeler language, and set
\(
K=f^{-1}(L).
\)
In order to  prove that \(K\) is Wheeler
we use the following characterization of Wheeler languages: a regular language
\(K\) is Wheeler if and only if every monotone sequence in
\(\operatorname{Pref}(K)\) is eventually constant modulo the
Myhill--Nerode equivalence of \(K\).

Let
\[
\alpha_0,\alpha_1,\alpha_2,\ldots
\]
be a monotone sequence in \(\operatorname{Pref}(K)\). We show that this
sequence is eventually constant modulo \(\equiv_K\).

Since the underlying DFA of \(T\) is Wheeler, the sets
\[
I_q=\{\alpha\in\Sigma^*:\delta(s,\alpha)=q\}
\]
are convex and totally ordered according to colexicographical order of the strings arriving in the states. Moreover from $K=f^{-1}(L)$ we obtain \[\pf{K}\subseteq \pf{dom(f)}.\] Since
there are only finitely many states, every monotone sequence can cross only
finitely many such intervals. Hence, after discarding a finite prefix of the
sequence, we may assume that there is a fixed state \(q\in Q\) such that
\(
\alpha_i\in I_q
\)
for every \(i\). 

By hypothesis,
\(
\omega(s,-)\restriction I_q
\)
is piecewise monotone. Hence there exists a finite convex partition
\(
I_q=C_1\cup\cdots\cup C_m\) with \(
C_1<\cdots<C_m,
\)
such that \(\omega(s,-)\) is monotone on each \(C_j\).   Therefore the sequence
\[
\omega(s,\alpha_0),\omega(s,\alpha_1),\omega(s,\alpha_2),\ldots
\]
is eventually monotone in \(\Gamma^*\). After discarding a further finite prefix, we may therefore assume that
\[
\omega(s,\alpha_0),\omega(s,\alpha_1),\omega(s,\alpha_2),\ldots
\]
is monotone in \(\Gamma^*\).

Now, since
\(
\alpha_i\in \operatorname{Pref}(K),
\)
there exists a word \(\beta_i\in\Sigma^*\) such that
\(
\alpha_i\beta_i\in K.
\)
Thus
\(
f(\alpha_i\beta_i)\in L.
\)
Moreover,
\[
f(\alpha_i\beta_i)
=
\omega(s,\alpha_i)\,\omega(q,\beta_i)\,t(r_i)
\]
for some final state \(r_i\in F\). Hence
\(
\omega(s,\alpha_i)\in \operatorname{Pref}(L).
\)

We have therefore obtained a monotone sequence
\(
\omega(s,\alpha_i)
\)
inside \(\pf{L}\). Since \(L\) is Wheeler, this sequence is
eventually constant modulo the Myhill--Nerode equivalence of \(L\). Hence,
after discarding a final finite prefix, we may assume that
\[
\omega(s,\alpha_i)\equiv_L \omega(s,\alpha_j)
\]
for all \(i,j\).

We now show that
\(
\alpha_i\equiv_K \alpha_j
\)
for all sufficiently large \(i,j\). Let \(\beta\in\Sigma^*\).
Since
\[
\delta(s,\alpha_i)=\delta(s,\alpha_j)=q,
\]
either we cannot read $\beta$ from $q$ and so $\alpha_i\beta, \alpha_j\beta \not \in K$, or
reading the continuation \(\beta\) from \(\alpha_i\) and from \(\alpha_j\)
 produces the same
output
\(
\omega(q,\beta)
\)
and reaches the same state
\(
r=\delta(q,\beta).
\)

If \(r\notin F\), then both \(\alpha_i\beta\) and \(\alpha_j\beta\) are outside
the domain of \(f\), and hence neither belongs to \(K\). If \(r\in F\), then
\[
f(\alpha_i\beta)
=
\omega(s,\alpha_i)\,\omega(q,\beta)\,t(r),
\]
and
\[
f(\alpha_j\beta)
=
\omega(s,\alpha_j)\,\omega(q,\beta)\,t(r).
\]
Since
\[
\omega(s,\alpha_i)\equiv_L \omega(s,\alpha_j),
\]
we get
\[
f(\alpha_i\beta)\in L
\iff
f(\alpha_j\beta)\in L.
\]
Equivalently,
\[
\alpha_i\beta\in K
\iff
\alpha_j\beta\in K.
\]

Thus, for every continuation \(\beta\in\Sigma^*\),
\[
\alpha_i\beta\in K
\iff
\alpha_j\beta\in K.
\]
Therefore
\(
\alpha_i\equiv_K \alpha_j.
\)

We have shown that every monotone sequence in \(\operatorname{Pref}(K)\) is
eventually constant modulo the Myhill--Nerode equivalence of \(K\). By the
characterization of Wheeler languages, \(K\) is Wheeler.
c
Since \(L\) was arbitrary, the function \(f=\|T\|\) is Wheeler-continuous.
 \hfill $\Box$

\end{proof}

\begin{lemma}  Let $T_f, T_g $ be two   Wheeler transducers, realizing the partial function $f:\Sigma^* \rightarrow \Gamma^*,g:\Gamma^*\rightarrow \Delta^*$ respectively, where the alphabet $\Sigma, \Gamma, \Delta$ are ordered. Then the composition $g \circ f$ is realized by a Wheeler transducer.  
\end{lemma}

\begin{proof} 
    Let
\[
T_f=(Q_f,s_f,\Sigma,\Gamma,\delta_f,\omega_f,F_f, t_f)
\]
and
\[
T_g=(Q_g,s_g,\Gamma,\Delta,\delta_g,\omega_g,F_g,  t_g)
\]
be the two transducers. 
We consider  the usual product construction $T_f\times T_g$ for transducers, realizing the composition $g\circ f$:

 \begin{itemize}
        \item the set of states is $Q_f\times Q_g$;
        \item the initial state is  $(s_f, s_{g})$;
        \item the input transition function is    \(\delta_{T_g\circ  T_f}((q,r),a)=(\delta_f(q,a), \delta_g(r,\omega_f(q,a)))\);
        \item the output transition  function is  \(\omega_{T_g\circ  T_f}((q,r),a)=\omega_g(r, \omega_f(q,a))  \);
        \item the set of final states is  $F_{T_g\circ  T_f}=\{(q,r):q\in F_f  \wedge \delta_g(r, t_f(q))\in F_g\}$;
        \item the terminal function is  \(t_{T_g\circ  T_f}(q,r)= \omega_g(r,t_f(q))t_g(\delta_g(r, t_f(q)))\).  
    \end{itemize}
Then, by induction on the length $|\alpha|$ of $\alpha\in \pf{dom(f)}$, it is possible to check   that the following holds:
\begin{equation}\label{eq:1}\delta_{T_g\circ  T_f}((q,r),\alpha)=(\delta_f(q,\alpha ), \delta_g(r,\omega_f(q,\alpha))),\end{equation} 
\begin{equation}\label{eq:2}\omega_{T_g\circ  T_f}((q,r),\alpha)=\omega_g(r,\omega_f(q,\alpha))  .\end{equation} 
Then it  easily  follows that $T_f\times T_g$ computes the function $g\circ f$. 
We trim   the  transducer $T_f\times T_g$ keeping only
those pairs \((q,r)\) from which some continuation leads to an input word in
\(dom(g\circ f)\). After this trimming, for every remaining pair
\((p,r)\), the   set of input words
\begin{align*}
  K_{p,r} &=
\{\alpha\in\Sigma^*~:~ \delta_{T_g\circ  T_f}((s_f,s_g),\alpha)=(p,r)\}\\
&=
\{\alpha\in\Sigma^*~:~ \delta_f(s_f,\alpha)=p
\text{ and }
\delta_g(s_g,\omega_f(s_f,\alpha))=r
\}
\end{align*}
 
is contained in
\[
P=\pf{{dom}(g\circ f)}.
\]

Equivalently,

\[
K_{p,r}
=
I_p\cap \omega_f(s_f,-)^{-1}(J_r), 
\]
where $I_p$ and $J_r$ are the set of input words arriving in $p$ and $r$ in $T_f$ and $T_g$, respectively.
Since   \(T_f\)  and  \(T_g\)  are Wheeler, each \(I_p\) and  \(J_r\) are convex, while \(K_{p,r}\) need not be so, since \(\omega_f(s_f,-)\) is piecewise monotone   on \(I_p\) and not necessarily monotone.  
However,   there exists a finite convex partition

\[
I_p=C_1\cup\cdots\cup C_m
\]

such that \(\omega_f(s_f,-)\) is monotone on each \(C_i\). This implies that, for each \(i,r\), the
set
\[
C_i\cap \omega_f(s_f,-)^{-1}(J_r)= (\omega_f(s_f,-)\restriction C_i)^{-1}(J_r)
\] is convex, 
being  the inverse image of a convex set
under a monotone function. Therefore \(K_{p,r}\) is a finite union
of convex sets \( C_i\cap \omega_f(s_f,-)^{-1}(J_r),
\) for states $(p,r)$. Moreover, notice that, for all $i,r$,  the function 
 \[\omega_{T_g\circ  T_f}((s_f,s_g),- )=\omega_g(s_g, \omega_f(s_f,-))  \]
 is piecewise monotone over \( C_i\cap \omega_f(s_f,-)^{-1}(J_r)\), being the composition of the two piecewise monotone functions  \[\omega_f(s_f,-))\restriction C_i\cap \omega_f(s_f,-)^{-1}(J_r) ~ \text{and} ~ \omega_g(s_g, -)\restriction J_r\]
 (see Lemma \ref{lem:comp_piecewise}).  However, in order to obtain a Wheeler transducer computing $g \circ f$ we cannot use  as new states the members  sets \( C_i\cap \omega_f(s_f,-)^{-1}(J_r)
\) of the $K_{p,r}$ convex decomposition,  because this partition would not, in general, be right invariant. Nevertheless, if a set has a convex partition realizing the piecewise monotonicity of a function, then any convex partition of the set would realize the piecewise monotonicity of the  function as well. Hence, for our purpose it will be enough to use any finite convex decomposition of $K_{p,r}$ which is right invariant.  In particular, as 
  states of the new transducer, we use  the convex sets belonging to the maximal convex components of some \(K_{p,r}\),    in the ordered set
\(
P=\pf{dom(g\circ f)}.
\)
These convex components are finite in number, because there are finitely many
pairs \((p,r)\), and each \(K_{p,r}\) is a finite union of convex sets.

In other words, we consider   the following equivalence relation $\sim$:
\begin{align*}
 \alpha \sim \alpha' \iff &\exists (p,r) \;\alpha , \alpha' \in K_{p,r} \; \wedge \;\\
 & \alpha, \alpha'\;\text{belong to the same maximal convex component of }\;  K_{p,r}   
\end{align*}

We next prove that $\sim$ is right invariant. 
Notice that, if \(
C\subseteq K_{p,r}
\) is a maximal convex component of \(K_{p,r}\) and  $a\in \Sigma$ then   \(Ca\cap P\)   is contained in  \(K_{p',r'}\)  where $(p,a,\omega_f(p,a),p')$ is a $T_f$ transition and   
\(\delta_g(r,\omega_f(p,a))=r' \) for some $r'$.
Moreover, \(Ca \cap P\) is convex. Indeed, appending the same final letter \(a\)
preserves the co-lexicographic order. Thus if \(Ca\cap P\neq \emptyset\) then \(Ca\cap P\) must be  contained in a unique maximal convex component of
\(K_{p',r'}\).

By construction, $\sim$ is refinement of $\sim_{T_f\times T_g}$ and hence we can apply Lemma \ref{lem:sim_quotient} obtaining a   quotient transducer $T_\sim$ still computing $g\circ f$.

It remains to check that this transducer is Wheeler.

The set of words reaching a state \(C\) of the refined automaton is exactly
\(C\). Since the states form a partition of \(P\) into pairwise disjoint convex
sets, they are totally ordered by the co-lexicographic order. Therefore, for any
two states \(C,D\), either every word of \(C\) precedes every word of \(D\), or
conversely. This is precisely the Wheeler order of the underlying DFA.

Finally, we verify the piecewise monotonicity condition on each new state.
Let
\(
C\subseteq K_{p,r}
\)
be one of the new states. For \(\alpha\in C\), the  output produced by
the composition is
\[
\omega_h(s_h,\alpha)
=
\omega_g(s_g,\omega_f(s_f,\alpha)).
\]

Since \(C\subseteq I_p\), the function
\(
\omega_f(s_f,-)\restriction C
\)
is piecewise monotone. Moreover, since \(C\subseteq K_{p,r} \subseteq   \omega_f(s_f,-)^{-1}(J_r)\), we have
\(
\omega_f(s_f,C)\subseteq J_r.
\)

By assumption,
\[
\omega_g(s_g,-)\restriction J_r
\]
is piecewise monotone. Since the composition of two piecewise monotone
functions is piecewise monotone, it follows that
\(
\omega_h(s_h,-)\restriction C
\)
is piecewise monotone.

Thus the refined product transducer computes \(g\circ f\), has a Wheeler
underlying DFA, and satisfies the   statewise piecewise monotonicity
condition.

Therefore the Wheeler class is closed under composition.
\hfill $\Box$
\end{proof}

\subsection{Minimization}
 In this section we will adapt Choffrut construction (Def.\ref{def:synct_congruence}, Def.\ref{def:Choffrut}) to the Wheeler class. 
 To obtain the corresponding syntactic congruence $\sim_f^c$  for Wheeler transducers, we need a convexity constraint on $\sim_f$, together with an order constraint on $\widehat f$. In the following, we denote by $[\alpha]_f$ the equivalence class of the word $\alpha$ with respect to $\sim_f$ and, for words $\mu, \mu' \in \pf{dom(f)}$,   we denote by $[\mu, \mu' ]^\pm$ the following interval in $(\pf{dom(f)}, \prec)$:
 \[
 [\mu, \mu' ]^\pm = 
 \begin{cases}
 [\mu, \mu']=\{\beta \in \pf{dom(f)}\;:\; \mu \preceq \beta \preceq \mu'\} \; \text{if}\; \mu \leq \mu'\\
  [\mu', \mu]=\{\beta \in \pf{dom(f)} \;:\; \mu' \preceq \beta \preceq \mu\}\;\text{if}\; \mu' \leq \mu 
 \end{cases}
 \]

\begin{definition}
Let $f:\Sigma^* \rightarrow \Gamma^*$ be a partial function. Then the  relation $\sim^c_f$ on $\pf{dom(f)}$ is defined as follows:
 
 \[
  \alpha \sim^c_f \alpha' ~~ \Longleftrightarrow  ~~
     [ \alpha, \alpha' ]^\pm \subseteq [\alpha]_f\; \wedge \;
      \widehat  f \; \text{ is piecewise monotone   over }  [ \alpha, \alpha' ]^\pm
 \]

\end{definition}

\begin{remark}\label{rem:eq_convex}
    Notice that $\sim^c_f$ is an equivalence relation. In particular, it is transitive: suppose $\alpha\sim^c_f\alpha'$ and $\alpha'\sim^c_f\alpha''$. Since  \([ \alpha, \alpha'' ]^\pm\subseteq [ \alpha, \alpha' ]^\pm \cup  [ \alpha', \alpha'' ]^\pm \)  we obtain that  \( [ \alpha, \alpha'' ]^\pm\subseteq [\alpha]_f\)  and \(\widehat  f \; \text{is piecewise monotone   over } [ \alpha, \alpha'' ]^\pm\). The last implication uses the fact that if \(D\) is convex and
\(D\subseteq C_1\cup C_2\), where \(f\) is piecewise monotone on both
\(C_1\) and \(C_2\), then \(f\) is piecewise monotone on \(D\).
Moreover, notice that the equivalence classes of $\sim^c_f$ are convex sets in $\pf{dom(f)}$. 
 
If $T$ is a transducer computing $f$,  $q$ is a $T$-state, and   $q=\delta(s,\alpha)$,  let  \(y=\bigwedge_{ \delta(q,\gamma)\in F } \omega(q, \gamma)t(\delta(q,\gamma)).\) Then, for all $\alpha \ \in I_q$ it holds: 
\begin{align*}\widehat f(\alpha)= \bigwedge_{\alpha\gamma\in  {dom(f)}}f(\alpha \gamma )&= \bigwedge_{\alpha\gamma\in  {dom(f)}} \omega(s,\alpha \gamma )t(\delta(s, \alpha\gamma))=\\
&=\omega(s, \alpha) \bigwedge_{\delta(q,\gamma)\in F}  \omega(q, \gamma)t(\delta(q,\gamma))=\omega(s, \alpha)  y,  \end{align*} where $y$ does not depend on $\alpha$ but only on $q$. 
It follows that

\begin{align*}
    &\omega \text{\; is (piecewise) monotone on \; } I_q    \iff \\ &\widehat f  \text{\; is (piecewise) monotone on \; } I_q    \iff \\ &  f  \text{\; is (piecewise) monotone on \; } I_q \text{\;(when $q$ is final)}.  
\end{align*}
\end{remark}

 \begin{lemma}\label{lem:right_invariance}
     Given a partial function  $f:\Sigma^* \rightarrow \Gamma^*$ the equivalence  $\alpha \sim_f^c \alpha'$   is right-invariant.
 \end{lemma}

\begin{proof}  If  $\alpha\sim_f^c\alpha'$  and $a\in \Sigma$, notice that $\alpha a \in \pf{dom(f)}$  iff $\alpha' a \in \pf{dom(f)}$. Hence we may suppose that both $\alpha a, \alpha'a \in \pf{dom(f)}$ and prove that  $\alpha a\sim_f^c\alpha' a$. 
For simplicity we only consider the case in which $\alpha \prec \alpha'$, the other case being   analogous.

    We begin by proving the first condition, that is $[ \alpha a, \alpha'a ]^\pm=[ \alpha a, \alpha'a ]  \subseteq [\alpha a]_f $.
   If $\beta\in [ \alpha a, \alpha'a ]^\pm =[ \alpha a, \alpha'a ] $  then   $\alpha a \prec \beta \prec \alpha'a$  and  $\beta = \gamma a$, for some $\gamma \in \Sigma^*$ with  $\alpha \preceq \gamma \preceq \alpha'$.  Since $\alpha \sim_f^c \alpha'$   we have   $[ \alpha , \alpha' ]\subseteq [\alpha]_f$ so that  $\alpha \sim_f \gamma$ holds. By Choffrut results we know  that $\sim_f$ is right invariant so $\alpha \sim_f \gamma \implies \forall a \in \Sigma^* \;\alpha a \sim_f \gamma a$. We can conclude that $ \alpha a \sim_f \beta$ holds.

   Next we prove that, under the hypothesis that $\alpha  \sim_f^c \alpha'  $, the function  $\widehat f$ is piecewise monotone over the interval 
   $[ \alpha a, \alpha'a ]$. From   $\alpha\sim_f^c \alpha'$  we know  that  $\widehat f$ is piecewise monotone  on $[\alpha, \alpha']$. Moreover, \([ \alpha a, \alpha'a ]=[\alpha, \alpha']a\cap \pf{dom(f)}\); hence, if 
   \([\alpha, \alpha']=C_1~ \cup~ \dots~ \cup~ C_k\) is a \([\alpha, \alpha']\)-convex partition   such that $\widehat  f$ is monotone over any  $C_i$s  and    \(D_i=C_ia\cap  \pf{dom(f)}\), then   \([\alpha, \alpha']a \cap \pf{dom(f)} =D_1 ~\cup ~ \dots~ \cup~  D_k\) is a  \([\alpha a, \alpha'a ]\)-convex partition. We next prove  that $\widehat f$ is monotone over the $D_i's$.   Indeed, by  Remark \ref{rem:eq_convex},  since $ C_i\subseteq [\alpha]_f$ there exists $x$ such that if $\beta =\gamma a  \in D_i$ with $\gamma \in C_i $ then $\widehat f(\beta)=\widehat f(\gamma)x $, so that $\widehat f$ inherits monotonicity on $D_i$ from monotonicity on $C_i$.   
\hfill $\Box$

\end{proof}

The congruence $\sim_f^c$ allows us to establish a Myhill--Nerode theorem for Wheeler  transducers. First we prove that, if $f$ is a function which is  computable via a Wheeler transducer,   the syntactic congruence $\sim_f^c$   has finite index.

 \begin{lemma}\label{lem:Wheeler_versus_finite}  If $f:\Sigma^*\rightarrow \Gamma^*$ is a function realized by a Wheeler transducer then $\sim_f^c$ has finite index  and
 $\widehat f$ is piecewise monotone over each $\sim_f^c$-class $[\alpha]_f^c$.
 \end{lemma}
 
\begin{proof}
    Let $T = (Q,s,\Sigma,\Gamma,\delta,\omega,F,t)$ be a Wheeler transducer   with $||T||=f$. Consider the  equivalence $\sim_T$   on $\pf{dom(f)}$  defined as in  (\ref{eq:simT}). 
 
We will  prove that $\sim_T$ is a refinement of $\sim^c_f$. Since $\sim_T $ has finite index, this will imply that $\sim^c_f$ has finite index as well. 

Suppose $\alpha \sim_T \alpha' $ and $\alpha \prec \alpha'$. We have to prove that $[\alpha, \alpha']\subseteq  [\alpha]_f$ and that  $\widehat f$ is piecewise monotone over $[\alpha, \alpha']$. Since $(Q,s,\Sigma,\delta,F)$ is a Wheeler DFA, the set $I_{\delta(s,\alpha)}$ is convex. Since $I_{\delta(s,\alpha)}\subseteq [\alpha]_f$, this implies $[\alpha, \alpha']\subseteq  [\alpha]_f$. 

Let    $q=\delta(s,\alpha)$. Since $T$ is a Wheeler transducer, we know that $ \omega$ is piecewise monotone over $I_q$. By Remark \ref{rem:eq_convex} we obtain that $\widehat f$ is also 
 piecewise monotone over $I_q$. Moreover $I_q$ is a convex set, and $\alpha, \alpha'\in I_q$, so $[ \alpha, \alpha']\subseteq I_q$. It follows that $\widehat f$    is piecewise monotone over $[\alpha, \alpha']$, as required. 
This concludes the proof that  $\sim_T$ is a refinement of $\sim_f^c$.  

To  prove that $\widehat f$ is piecewise monotone over each $\sim_f^c$-class $[\alpha]_f^c$, notice that, as proved before, $\widehat f$ is piecewise monotone over each $I_q$, which are  convex sets, and each  $\sim_f^c$-class is a finite union of some $I_q$.  
\end{proof}
 \hfill $\Box$

\begin{lemma} \label{lem:finite_versus_Wheeler} Let  $f:\Sigma^* \rightarrow \Gamma^*$. If    the equivalence \(\sim^c_f\) has finite index and $\widehat f$ is piecewise monotone over each class $[\alpha]_f^c$ then $f$ is a Wheeler function.
\end{lemma}

\begin{proof} Since $\sim_f^c$ has finite index and it is a refinement of $\sim_f$, we have that $\sim_f$ has finite index.  Consider  the   refinement \(\sim'\) of  \(\sim^c_f\)  isolating $\epsilon$ in a single class:
\[\alpha\sim'\alpha' \iff (\alpha, \alpha'\neq \epsilon \land \alpha \sim_f^c \alpha')  \lor (\alpha=\alpha'=\epsilon)
\] 
Then $\sim'$ is   right invariant (because $\sim_f^c$ is so, by Lemma \ref{lem:right_invariance}). Moreover, $\sim'$ is a refinement of the relation $\sim_{T_f}$ corresponding to the Choffrut    transducer  \(T_f=(Q_f,s_f,\delta_f,\omega_f,F_f,t_f)\). We can hence apply     Lemma \ref{lem:sim_quotient}. The quotient transducer $T_{\sim'}=(Q,s,  \delta, \omega , F,t)$, has the $\sim'$ equivalence classes as states, computes the function $f$ and, for    $\alpha\in \pf{dom (f)} $,  we have:
\begin{align*}
\delta(s,\alpha)&=[\alpha]_{\sim'}\\
   I_{[\alpha]_{\sim'}}&=[\alpha]_{\sim'}   \\ 
   \omega (s,\beta)&=\omega_{T_f}(s_f, \beta).
\end{align*}    
 
By Remark \ref{rem:eq_convex} the \(\sim_f^c\)-classes are convex and this implies, since $\epsilon$ is minimum in the colexicographical order, that  the same is true for the $\sim'$ equivalence classes which are   the states of the transducer $T_{\sim'}$. It follows that  these states can be ordered
colexicographically and  $A_{T_{\sim'}}$ is Wheeler.

\medskip

As for the  second Wheeler condition, notice that   by hypothesis   \(\widehat f\) is piecewise monotone   over the $\sim_f^c$ classes. Moreover  $\sim_f^c$ classes  coincide with $\sim'$ classes,  except for   $[\epsilon]_f^c$    which is splitted into  $\{\epsilon\}$ and  $[\epsilon]_f^c\setminus \{\epsilon\}$.  Since \(I_{[\alpha]_{\sim'}}=[\alpha]_{\sim'}\),  it is then clear that \(\widehat f\) is piecewise monotone over the sets $I_{[\alpha]_{\sim'}}$. 
By Remark \ref{rem:eq_convex} we obtain that $\omega$ is piecewise monotone over   $I_{[\alpha]_{\sim'}}$. This concludes the proof that $T_{\sim'}$ is a Wheeler transducer. Since  $T_{\sim'}$ computes $f$, the Lemma follows. 
 \hfill $\Box$
\end{proof}

Using Lemma \ref{lem:Wheeler_versus_finite}  and Lemma \ref{lem:finite_versus_Wheeler}  we can finally prove our   characterization of Wheeler transducers:

\begin{theorem}\label{thm:subseq_wheeler_finite_index} A partial function  $f:\Sigma^* \rightarrow \Gamma^*$ is computed  by a   Wheeler transducer if and only if \(\sim^c_f\) has finite index and $\widehat f$  is piecewise monotone over each  \(\sim^c_f\)-class. 
\end{theorem}
 
Using the previous theorem we may also obtain a characterization of Wheeler functions   that highlights the monotonicity property directly on the function $f$, rather than on $\widehat f$:

\begin{theorem}\label{transducerFree}
  A  sequential function $f:\Sigma^*\rightarrow \Gamma^*$ is a Wheeler function if and only if  every  monotone sequence  in $\pf{dom(f)}$  eventually lies in a single  $\sim_f$-class  and $f$ is piecewise  monotone.
\end{theorem}

\begin{proof} 
    ($\implies$) 
    Let $f:\Sigma^*\to\Gamma^*$ be a Wheeler function and let \[T=(Q,s,\Sigma,\Gamma,\delta,\omega,F,t)\] be a Wheeler transducer realizing $f$. 
    Consider a  monotone sequence  $(\alpha_i)_{i\in\mathbb N}$   in $\pf{dom(f)}$.
    Since   $T$ is a Wheeler DFA, the sets $I_q$ for $q\in Q$ are convex and form an ordered partition of $\pf{dom(f)} $. Hence   there exist $n_1\geq 0$ and a state $q\in Q$ such that $\delta(s,\alpha_i)=q$ for all $i\geq n_1$.  This implies  $\alpha_i\sim_f\alpha_j$ for all $i,j\geq n_1$  and  proves that $(\alpha_i)_{i\in\mathbb N}$ ends up, eventually, in a single $\sim_f$-class.

   Moreover,  since $T$ is Wheeler, the function   
   $\omega(s,-)$ is piecewise monotone over each $I_q$.   If $q\in F$  and $\alpha \in I_q$ then we have $f(\alpha)=\omega(s, \alpha)t(q)$. Hence for $q\in F$ we have that  $f$ is piecewise monotone on   $I_q$, implying that $f$ is piecewise monotone over $dom(f)=\bigcup_{q\in F} I_q$. 
   
\medskip

($\impliedby$) Here we apply Theorem \ref{thm:subseq_wheeler_finite_index}, proving that   the  congruence relation $\sim_f^c$    has finite index and that $\widehat f$ is piecewise monotone over its classes.  Suppose that $\sim_f^c$ does not have a finite number of classes, then, since $f$ is a sequential function and   $\sim_f$ has finite index,      there must  be   a class of $\sim_f$ containing  an infinite number of classes of $\sim_f^c$. If  we call these classes $[\alpha_1]_f^c,[\alpha_2]_f^c,\dots$ we know that   $\alpha_i \not\sim_f^c \alpha_j$ and $\alpha_i\sim_f \alpha_j$,  $\forall i\neq j$.

  Consider a monotone  subsequence    $(\alpha'_i)_{i \in \mathbb{N}}$  of the sequence  $(\alpha_i)_{i \in \mathbb{N}}$ and suppose     that  $(\alpha'_i)_{i \in \mathbb{N}}$ is increasing (the decreasing case is analogous).  Since, for all $i \in \mathbb{N}$, we have    $\alpha'_i \not\sim_f^c \alpha'_{i+1}$, it must be the case that  
    \[\forall i  \left ([\alpha'_i, \alpha'_{i+1}]\not \subseteq [\alpha'_i]_f ~\lor ~\widehat f ~\text{is not piecewise monotone on} ~ [\alpha'_i, \alpha'_{i+1}]\right )\]

    If, for some $i$,  the first alternative fails, then the interval $[\alpha'_i, \alpha'_{i+1}]$ is contained in one
\(\sim_f\)-class and it must be the case that    $\widehat f $  is not piecewise monotone on $[\alpha'_i, \alpha'_{i+1}]$.  Since $[\alpha'_i, \alpha'_{i+1}]\subseteq [\alpha'_i]_f $ and $[\alpha'_i, \alpha'_{i+1}]\subseteq \pf{dom(f)}$,  there exist  $\beta,x$ such that $[\alpha'_i, \alpha'_{i+1}]\beta \subseteq dom(f)$ and $f(\gamma\beta)=\widehat f(\gamma) x$, for all $\gamma\in [\alpha'_i, \alpha'_{i+1}]$. This implies that $f$ is not piecewise monotone on $[\alpha'_i, \alpha'_{i+1}]\beta \subseteq dom(f)$, a contradiction because $f$ is piecewise monotone over $dom(f)$ and $[\alpha'_i, \alpha'_{i+1}]\beta $ is a convex subset of  $dom(f)$. 

Hence, we must have that $[\alpha'_i, \alpha'_{i+1}]\not \subseteq [\alpha'_i]_f$, for all $i$, that is,    \[\forall i  ~\exists \beta_i \in \pf{dom(f)} ~( \alpha'_i\prec \beta_i \prec \alpha'_{i+1} \land (\alpha'_i \not \sim_f \beta_i).\] Then the monotone sequence $\alpha'_{0} \prec  \beta_{ 0} \prec \alpha'_{ 1} \prec  \beta_{1} \prec  \dots$  does not  end  up, eventually, in the same $\sim_f$ class, contradicting the hypothesis. 

Finally, we want to prove that $\widehat f$ is piecewise monotone over each   $\sim_f^c$-class  $[\alpha]_f^c$. We know that  $[\alpha]_f^c\subseteq [\alpha]_f$, hence there exists $\beta$ with $[\alpha]_f^c\beta\subseteq dom(f)$, and there exists $x$ such that
$f(\gamma \beta)=\widehat f (\gamma)x$, for all $\gamma\in [\alpha]_f^c$.  Since $[\alpha]_f^c\beta$ is convex on $dom(f)$ and by hypothesis $f$ is piecewise monotone on $dom(f)$, we have that $f$ is piecewise monotone on $[\alpha]_f^c\beta$, and, reasoning as before, we obtain that $\widehat f$ is piecewise monotone on $[\alpha]_f^c$. 
    \hfill $\Box$
\end{proof}

\begin{remark}\label{rem:not_all}
In the previous sections we proved that a function computed by a   Wheeler transducer is always   Wheeler-continuous. 
  There are two independent requests for membership in  the class of Wheeler transducers: 1) the underlying DFA must be Wheeler; 2) the output  function $\omega(s,-)$ must be piecewise monotone over each $I_q$. If we drop one of the two conditions, we may lose the Wheeler continuity   of the computed function. E.g.,   the transducer $T$ over $\Sigma=\{a\}, \Gamma=\{a \}$ with transitions 
$(q_0,a,a,q_1), (q_1,a,a, q_0), t(q_0)=\epsilon$  computes the identity function over $dom(||T||)$, hence it satisfies the second  condition but it is not Wheeler-continuous since   $dom(T)=Id^{-1}(\Sigma^*)=a^{2n}$, and $a^{2n}$ is not Wheeler. 

Moreover, if we drop the second condition, we may still lose Wheeler continuity, as the following transducer shows: 
Let
\(
\Sigma=\Gamma=\{0,1\},
\) with $0\prec 1$.
Define the total function
\(
f:\Sigma^*\to \Gamma^*
\)
by letterwise complementation:
\(
f(w)=\overline w,
\)
where
\(
\overline 0=1,
\overline 1=0.
\) 
Then \(f\) is subsequential, realized by the one-state transducer with transitions $(q,0,1, q), (q,1,0,q)$ whose underlying DFA is Wheeler but \(f\) is not piecewise monotone and is not Wheeler-continuous.
$f$ is not piecewise monotone because along    the  strict increasing sequence \[
\alpha_{2k}=001^k,
\qquad
\alpha_{2k+1}=1^{k+1}
\qquad(k\ge 0) 
\]
 the outputs alternate:
\[
11\succ 0\prec110\succ00\prec1100\succ 000\prec\cdots.
\]
 
In order to prove that \(f\) is not Wheeler-continuous, consider the  Wheeler language
\[
L=0^*\cup 110^*1\subseteq \Gamma^*
\]

Since \(f\) exchanges \(0\) and \(1\), we have
\[
f^{-1}(L)=1^*\cup 001^*0,
\]
and this language   is not Wheeler.

Hence the proposed  class of Wheeler transducers is, in a sense minimal. 
However, the Wheeler class does not compute    all subsequential  Wheeler-continuous functions, even among monotone functions,  as the following example shows.
Let 
\(
\Sigma=\{a\},
\Gamma=\{a\},
\)
and define 
\[
f(a^n)=a^{\lfloor n/2\rfloor}.
\]
Then 
\[
f(\varepsilon)=\varepsilon,\quad
f(a)=\varepsilon,\quad
f(aa)=a,\quad
f(aaa)=a,\quad
f(aaaa)=aa,\quad \ldots
\]

The function \(f\) is realized by the subsequential transducer 
\[
s \xrightarrow{a\mid \varepsilon} q ,
\qquad
q  \xrightarrow{a\mid a} s,
\]
where all states are final and  $t(s)=t(q)=\epsilon$.
Moreover  \(f\) is monotone nondecreasing.
 Since Wheeler  language over a unary alphabet    are only the finite or cofinite sets, it is also easy to check that \(f\) is Wheeler-continuous.

We now show that \(f\) is not computed by any Wheeler transducer.
Suppose, by way of a contradiction,   that  \(f\) is computed by a Wheeler transducer. 
Then, there exists $N$ and a state $q$ such that $\delta(s,a^n)=q$, for all $n\geq N$.  
Hence, the transducer must contain a transition 
\((q,a,\alpha, q)\)
with  \(\alpha\in\Gamma^*\). This means that,  for all  \(n\geq N\), we have 
\(\omega(s,a^{n+1})=\omega(s,a^n)\alpha
\)
and, if  $t(q)=\alpha'$, we should have $f(a^{N+k})=\omega(s,a^N)\alpha^k\alpha'$. 
However, this 
means that, over the sequence \(a^{N+k}\), the function $f$ should either be constant (if $\alpha=\epsilon$), or grow w.r.t. the length at each step and  this is not true for the function \(f(a^n)=a^{\lfloor n/2\rfloor}\), a contradiction.    
\end{remark}
 
\section{Conclusions and Future Work}

The study initiated here paves the way for a broader research program
on the structural theory of Wheeler transducers.

Several fundamental problems concerning decidability and complexity remain unresolved. The most immediate is the decidability of membership of a   transducer in the class of Wheeler transducers and the complexity of such a decision procedure. A related problem is that of Wheeler realizability: given an arbitrary   transducer, can one decide whether an equivalent Wheeler transducer exists? The complexity of the minimization procedure induced by $\sim_f^c$ likewise remains open.

A further direction for future work is to generalize the notion of Wheeler transducers beyond the present setting. In particular, it would be interesting to develop a hierarchy of   transducers mirroring the hierarchy of regular languages induced by the deterministic width of their co-lexicographic order (\cite{cdpp23}). Such a hierarchy could provide a finer classification of   transducers, relating their structural complexity to the degree to which their underlying automata depart from the Wheeler case. This may also lead to new algorithmic techniques for indexing and processing transductions with bounded co-lexicographic width.

Beyond the deterministic case, a natural extension is the study of Wheeler relations, analogously to rational relations in the regular setting, obtained by dropping the determinism requirement. In order to proceed in this direction, it is necessary to develop a theory of nondeterministic Wheeler transducers and to understand which properties of rational relations survive in the Wheeler setting.
Finally, a logical characterization of Wheeler functions, analogous to the
connection between regular functions and monadic second-order logic, is another
direction for future work.

  \subsubsection*{Acknowledgments}
 We thank the anonymous referee   for a careful reading of the manuscript and for
 pointing out a serious  flaw in a previous  definition of $\sim_f^c$.

\section*{Declaration on Generative AI}
During the preparation of this work, the authors used ChatGPT to check grammar and spelling, rephrase and reword text, and create examples. After using this tool, the authors reviewed and edited the content as needed and take full responsibility for the publication’s content.
 
%
%
\bibliographystyle{splncs04}
\bibliography{biblio}

\end{document}